\documentclass[11pt]{article}
\usepackage[text={7.2in,9.6in}]{geometry} 
\usepackage{graphicx}
\usepackage{amsmath}
\usepackage{amssymb}
\usepackage{epstopdf}
\usepackage{graphicx}
\usepackage{booktabs}
\usepackage{fancyhdr}
\graphicspath{ {images/} }
\usepackage{amsthm} 

\usepackage{enumerate}
\usepackage{setspace}	

\usepackage[all]{xy}
\newtheorem{lemma}{Lemma}

\begin{document}
\onehalfspacing
\begin{center}
{\LARGE \textbf{Two Error Bounds of Imperfect Binary Search}}
\vspace{5mm}

{\Large Haoze Wu}
\end{center}
\noindent


Suppose we know that an object is in a sorted table and we want to determine the index of that object. To achieve this goal we could perform a binary search\cite{Knuth:1998:ACP:280635}. However, suppose it is time-consuming to determine the relative position of that object to any other objects in the table. In this scenario, we might want to resort to an incomplete solution: we could device an algorithm that quickly predicts the result of comparing two objects, and replace the actual comparison with this algorithm during a binary search.

The question then is how far away are the results yielded by the imperfect binary search from the correct answers. We present two lemmas about the expected error of a imperfect binary search that  goes the wrong direction with a fixed probability.

\begin{lemma}
The expected error for a binary search that goes the wrong direction with probability $\epsilon$  on a size $n$ table is upper-bounded by $\epsilon * n$.
\end{lemma}
\begin{proof}
Let $a(n)$ be the expected error for the binary search method on a table with size $n$. At each iteration, the binary search method goes the wrong direction with probability $\epsilon$. 
	
	Consider the first iteration, where the binary search method decides whether to set $\frac{n}{2}$ as the new upper-bound or the new lower-bound. With probability $1-\epsilon$ the decision is correct. In this case, the expected error reduce to $a(\frac{n}{2})$. On the other hand, with probability $\epsilon$, the binary search method makes the wrong decision. In the case, the expected error is at most $\frac{n}{2} + a(\frac{n}{2})$. Therefore, the upperbound of the expected error, $a(n)$, is equal to  $(1-\epsilon) * a(\frac{n}{2}) + \epsilon * (\frac{n}{2} + a(\frac{n}{2})) = a(\frac{n}{2}) + \epsilon * \frac{n}{2}$.
	
	Similarly, $a(\frac{n}{2})  = a(\frac{n}{4}) + \epsilon * \frac{n}{4}$, and so forth. 
	
	Therefore, 
	
	\centering
	\[\arraycolsep=1.4pt\def\arraystretch{2.2}
	\begin{array} {lcl} a(n) &=&  \epsilon * \frac{n}{2} +  \epsilon * \frac{n}{4} +  \epsilon * \frac{n}{8} +  \epsilon * \frac{n}{16} + ... \\ & = & \epsilon*(\frac{n}{2} + \frac{n}{4}+ \frac{n}{8} + \frac{n}{16} + ...)\\ &=&  \epsilon*n \end{array}.
		\]
\end{proof}

\begin{lemma}
The expected error for a binary search that goes the wrong direction with probability $\epsilon$  on a size $n$ table on average is $\frac{\epsilon n (0.5 + \epsilon)}{1 + \epsilon}$.
\end{lemma}
	\begin{proof}
	Let $b(n)$ be the average expected error for the binary search method on a table with size $n$. Let $a(n)$ be the upper-bound of the expected error for the binary search method on a table with size $n$. From proposition 1, $a(n) = \epsilon * n$. At each iteration, the binary search method goes the wrong direction with probability $\epsilon$. 

	Consider the first iteration, where the binary search method decides whether to set $\frac{n}{2}$ as the new upper-bound or the new lower-bound. With probability $1-\epsilon$ the decision is correct. In this case, the expected error reduce to $e(\frac{n}{2})$. If
	
	On the other hand, with probability $\epsilon$, the binary search method makes the wrong decision. In the case, the expected error guarantees $\frac{n}{4}$ error. Additionally, the later binary search would produce an error of $a(\frac{n}{2})$. That is, in this case, the expected error is $\frac{n}{4} + a(\frac{n}{2}) = \frac{n}{4} + \epsilon * \frac{n}{2}$. Therefore, the expected error, $b(n)$, is equal to  $\epsilon * (\frac{n}{4} + \epsilon * \frac{n}{2}) + (1-\epsilon) * b(\frac{n}{2})$.
	
	Similarly, $b(\frac{n}{2})  = \epsilon * (\frac{n}{8} + \epsilon * \frac{n}{4}) + (1-\epsilon) * b(\frac{n}{4})$, and so forth. 
	
	Therefore, 
	
	\centering
	\[\arraycolsep=1.4pt\def\arraystretch{2.2}
	\begin{array} {lcl} b(n) &=&  \epsilon * (\frac{n}{4} + \epsilon * \frac{n}{2}) + (1-\epsilon)  [\epsilon * (\frac{n}{8} + \epsilon * \frac{n}{4}) + (1-\epsilon)  [\epsilon * (\frac{n}{16} + \epsilon * \frac{n}{8}) + (1-\epsilon) *...]] \\ & = & \epsilon n(\frac{1}{4} + \frac{1-\epsilon}{8} + \frac{(1-\epsilon)^2}{16} + ...) + \epsilon^2n(\frac{1}{2} + \frac{1-\epsilon}{4} + \frac{(1-\epsilon)^2}{8} + ... ) \\ &=& (\frac{\epsilon n}{4} + \frac{\epsilon^2 n}{2}) \sum\limits_{i=0}^\infty (\frac{1-\epsilon}{2})^i \\ &=& (\frac{\epsilon n}{4} + \frac{\epsilon^2 n}{2}) * \frac{1} {1-\frac{1-\epsilon}{2}} \\ &=&  \frac{1}{2}\epsilon n (\frac{1}{2} + \epsilon) * \frac{2}{1 + \epsilon} \\ &=& \frac{\epsilon n (0.5 + \epsilon)}{1 + \epsilon}\end{array}.
	\]
	\end{proof}
\bibliographystyle{abbrv}
\bibliography{sigproc}

\begin{thebibliography}{1}

\bibitem{Knuth:1998:ACP:280635}
D.~E. Knuth.
\newblock {\em The Art of Computer Programming, Volume 3: (2Nd Ed.) Sorting and
  Searching}.
\newblock Addison Wesley Longman Publishing Co., Inc., Redwood City, CA, USA,
  1998.

\end{thebibliography}

\end{document}